\documentclass[letterpaper,11pt]{article}
\usepackage{amsmath,amsfonts,amsthm}
\usepackage[margin=1in]{geometry}
\usepackage{enumitem}
\usepackage{tikz}
\tikzstyle{res}=[draw=black, fill=white, minimum size=1.3em, font={\tiny}]

\usepackage{thm-restate}

\title{Pure Nash Equilibria in Resource Graph Games}

\newcommand{\C}{\mathcal{C}}
\newcommand{\R}{\mathbb{R}}
\newcommand{\N}{\mathbb{N}}
\newcommand{\Z}{\mathbb{Z}}
\usepackage{bm}

\newcommand{\B}{{\mathcal B}}
\renewcommand{\vec}[1]{\mathbf{#1}}
\newtheorem{theorem}{Theorem}
\newtheorem{definition}{Definition}
\newtheorem{lemma}{Lemma}

\newtheorem{example}{Example}

\author{Tobias Harks\footnote{Augsburg University, Universit\"atsstra\ss e 14, 86159 Augsburg, Germany.	\texttt{tobias.harks@math.uni-augsburg.de}} \and Max Klimm\footnote{Technische Universit\"at Berlin, Stra\ss e des 17.~Juni 136, 10623 Berlin, Germany. \texttt{klimm@tu-berlin.de}} \and Jannik Matuschke\footnote{KU Leuven, Naamsestraat 69, 3000 Leuven, Belgium. \texttt{jannik.matuschke@kuleuven.be}}}

\begin{document}
\maketitle

\begin{abstract}
This paper studies the existence of pure Nash equilibria in resource graph games, which are a general class of strategic games used to succinctly represent the players' private costs. There is a finite set of resources and the strategy set of each player corresponds to a set of subsets of resources. The cost of a resource is an arbitrary function that depends on the load vector of the resources in a specified neighborhood. As our main
result, we give complete characterizations of the cost functions guaranteeing the existence of pure Nash equilibria for weighted and unweighted players, respectively.
\begin{enumerate}
\item For unweighted players, pure Nash equilibria are guaranteed to exist for any choice of the players' strategy space if and only if the cost of each resource is an arbitrary function of the load of the resource itself and linear in the  load of all other resources where the linear coefficients of mutual influence of different resources are symmetric. This implies in particular that for any other cost structure there is a resource graph game that does not have a pure Nash equilibrium.
\item For weighted games where players have intrinsic weights and the cost of each resource depends on the aggregated weight of its users, pure Nash equilibria are guaranteed to exist if and only if the cost of a resource is linear in all resource loads, and the linear factors of mutual influence are symmetric, or there is no interaction among resources and the cost is an exponential function of the local resource load. 
\item  For the special case that the players' strategy sets are matroids, we show that pure Nash equilibria exist under a local monotonicity property, even when cost functions are player-specific. We point out an application of this result to bilevel load balancing games, which are motivated by the study of network infrastructures that are resilient against external attackers and internal congestion effects.
\item  Finally, we discuss the computational complexity of deciding whether a given strategy profile is a pure Nash equilibrium and derive hardness results for network routing games and matroid games, respectively.
\end{enumerate}
\end{abstract}

\newpage

\section{Introduction}
Multi-agent systems are characterized by the intricate interplay of the different and sometimes conflicting self-interests of a large number of independent individuals. In order to study the effects of selfish behavior on the overall state of these systems game-theoretic solution concepts are used, most notably the concept of a Nash equilibrium.
Important questions for the analysis of multi-agents systems are, thus, under which conditions Nash equilibria exists and how they can be computed.
For systems with a large number of players (as they frequently appear in multi-agent systems modelling economic, traffic, or telecommunication applications), the representation of the games becomes an important issue. For illustration, consider a system with $n$ agents, each with $m$ strategies. Encoding the payoffs of each agent in each of the $m^n$ strategy profiles requires the encoding of $nm^n$ rational number which is impractical even for modest sizes of $n$ and $m$.
Fortunately, for many multi-agents systems that arise from practical application, the agents' payoff have additional structure that allows for a succinct representation of the payoffs. Examples include extensive form games, congestion games (Rosenthal~\cite{Rosenthal73a}), graphical games (Kearns et al.~\cite{KearnsLS01}), action graph games (Jiang et al.~\cite{JiangLB11}), and local effect games (Leyton-Brown and Tenneholtz~\cite{Leyton-BrownT03}).

A general class of games that includes several of the specific classes of games above is the class of \emph{resource graph games} introduced by Jiang et al.~\cite{JiangCL17}. In a resource graph game, we are given a finite set $N = \{1,\dots,n\}$ of players and a finite set $R = \{1,\dots,m\}$ of resources. The strategy set available to player~$i$ is a set $X_i \subseteq \{0,1\}^m$ with a succinct representation.\footnote{Jiang et al.\ consider a polytopal representation, but the exact specifics how the set $X_i$ is represented is less important in this work, as we focus on existence of equilibria and less on computational aspects.}
Given a strategy profile $x = (\vec x_1,\dots, \vec x_n)$, let $\vec x = \sum_{i \in N} \vec x_i \in \R_{\geq 0}^m$ denote the configuration profile representing the total number of players using each resource in strategy profile $x$. Then, the private cost of player~$i$ is defined as 
\begin{align*}
\pi_i(x) =   \vec x_i^\top \vec c(\vec x) = \sum_{r \in R} x_{i,r}\, c_r(\vec x) \quad \text{ for all $i \in N$},
\end{align*}
where $\vec c : \R^m_{\geq 0} \to \R^m$ is an arbitrary function.
For most applications, the function $\vec c : \R^m_{\geq 0} \to \R^m$ itself has a succinct representation of the following form. For each resource~$r$, let $B_r \subseteq R$ be an arbitrary subset of neighbors of $r \in R$ and assume that the function $\vec c : \R_{\geq 0}^m \to \R^m$, $\vec x \mapsto (c_1(\vec x), \dots, c_m(\vec x))$ has the property that for every resource~$r \in R$ the cost $c_r$ depends only on the configuration profile of the resources in $B_r$, i.e., $c_r(\vec x) = c_r(\vec y)$ for all $\vec x,\vec y \in \R_{\geq 0}^m$ with $x_s = y_s$ for all $s \in B_r$.
A graphical illustration of such a game is obtained by the graph that has the vertex set $R$ and a directed edge from $s$ to $r$ if and only if $s \in B_r$. This is the intuition behind the name \emph{resource graph games}, see Fig.~\ref{fig:resource-graph-games} for an illustration.

\begin{figure}
\centering
 \begin{tikzpicture}[scale=1.8]
      
    \begin{scope}
      \draw[rounded corners, fill=blue!20, opacity=0.5] (0.7, -0.3) -- ++(0, 1.6) -- ++(1.6, 0) -- 
        ++(0, -0.6) -- ++(-1, 0) -- ++(0, -1) -- cycle;
    \end{scope}
      
    \begin{scope}[xscale=-1, xshift=-5cm]
      \draw[rounded corners, fill=cyan!20, opacity=0.5] (0.7, -0.3) -- ++(0, 1.6) -- ++(1.6, 0) -- 
        ++(0, -0.6) -- ++(-1, 0) -- ++(0, -1) -- cycle;
    \end{scope}
      
    \begin{scope}
      \draw[rounded corners, fill=red!20, opacity=0.5] (1.8, -0.3) -- ++(0, 1.9) -- ++(2.4, 0) -- 
        ++(0, -0.9) -- ++(-0.4, 0) -- ++(0, 0.7) -- ++(-1.6, 0) -- ++(0, -1.7) -- cycle;
    \end{scope}
      
    \begin{scope}[xscale=-1, xshift=-5cm]
      \draw[rounded corners, fill=orange!20, opacity=0.5] (1.8, -0.3) -- ++(0, 2.2) -- ++(2.4, 0) -- 
        ++(0, -1.2) -- ++(-0.4, 0) -- ++(0, 1) -- ++(-1.6, 0) -- ++(0, -2) -- cycle;
    \end{scope}
    
    \node at (1.5, 1) {$\vec x_{1}$};
    \node at (3.5, 1) {$\vec x'_{1}$};
    
    \node at (4, 1.45) {$\vec x_{2}$};
    \node at (1, 1.5) {$\vec x'_{2}$};
    
    \foreach \x in {1,2,3,4}
    {
        \node[res] (r\x) at (\x,1) {$r_{\x}$};
        
        \node[res] (s1) at (1,0) {$r_{5}$};
        \node[res] (s2) at (2,0) {$r_{6}$};
        \node[res] (s3) at (3,0) {$r_{7}$};
        \node[res] (s4) at (4,0) {$r_{8}$};

        \draw[black!40, <->, thick] (r\x) -- (s\x);

    }
    \draw[black!40, ->, thick] (r1) -- (s2);
     \draw[black!40, <->, thick] (r3) -- (s4);
      \draw[black!40, ->, thick] (s3) -- (r4);
       \draw[black!40, ->, thick] (r3) -- (s1);
       
       \draw[black!40, ->, thick] (s3) -- (s4);
       \draw[black!40, ->, thick] (r3) edge[bend right=45] (r1);
    
  \end{tikzpicture}
  \caption{An example of a resource graph game with resource set $R=\{r_1,\dots, r_8\}$ visualizing the non-separable
  effects among resources w.r.t the cost function.
  The white rectangles represent the $8$ resources, the (bi)directed edges represent the 
  non-separability of costs, that is, a directed edge from node $r_j$ to node $r_i$
  indicates that the function value $ c_{r_i}(\vec x)$ depends on the entry $x_{r_j}$, or equivalently, $r_j\in B_{r_i}$. The colored subsets of resources represent the strategies  $\vec x_1, \vec y_1$ and $\vec x_2, \vec y_2$ of the two players $\{1,2\}$.
  }\label{fig:resource-graph-games}
  \end{figure}
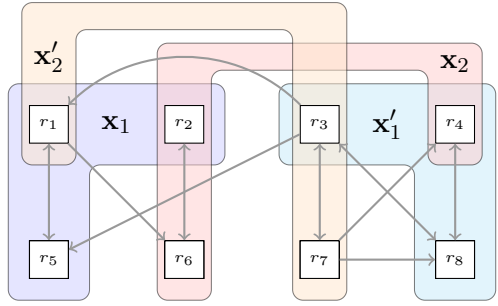

Below, we illustrate classes of games that are special cases of resource graph games. We start with the class of unweighted congestion games introduced by Rosenthal~\cite{Rosenthal73a} as a model for road traffic and for production with demand-dependent costs.

\begin{example}[Unweighted congestion games]
When the neighborhood of each resource~$r$ contains only $r$, i.e., $B_r = \{r\}$ for all $r \in R$, the cost of each resource depends only on the number of players using it. We then obtain the class of (unweighted) congestion games as a special case of resource graph games. 
\end{example}

Another example of a subclass of resource graph games is the class of local effect games as introduced by Leyton-Brown and Tennenholtz~\cite{Leyton-BrownT03}. Compared to unweighted congestion games, they are less general in terms of the players' strategies since only singleton strategies are allowed; in terms of the cost structure on the resources, they are more general since the cost of a resource may also depend on the load of other resources. We call a resource graph game a \emph{singleton game}, if $\sum_{r \in R} x_{i,r} = 1$ for all $\vec x_i = (x_{i,r})_{r \in R} \in X_i$ and all $i \in N$.

\begin{example}[Local effect games]
Local effect games are singleton resource graph games, where for every resource~$r$, there is a function
$f_r : \R_{\geq 0} \to \R$ and for every pair of resources $r,s \in R$ such that $s \in B_r$, there is a function $f_{r,s} : \R_{\geq 0} \to \R$ such that $c_r(\vec x) = f_r(x_r) + \sum_{s \in B_r} f_{s,r}(x_s)$ .
\end{example}

The following class of action graph games is a generalization of local effect games. They are introduced by Jiang et al.~\cite{JiangLB11} and generalize local effect games as they allow an arbitrary functional dependence of the cost of a resource on the load of all other resources, yet, they are a sublcass of resource graph games. Jiang et al.~discuss applications in modelling location games, congestion games, and anonymous games, but the class of games is universal as any strategic game can be represented as an action graph game. Thompson and Leyton-Brown~\cite{ThompsonL17} show how to use action graph games to compute equilibria in position auctions.

\begin{example}[Action graph games]
The class of action graph games is equivalent to the class of singleton resource graph games.
\end{example}

We now introduce a new class of games related to security games with congestion effects.
Consider a load balancing setting where players choose one resource out of a set of resources. After observing the realized loads, a follower attacks the resources with maximum loads causing additional disutilities for the players choosing the attacked resources. Attacks may be thought of as either being actual attacks by a malicious player or as controls by a central authority to counter tax or fare evasion (see Correa et al.~\cite{CHKM17} for a related mathematical model of fare evasion without any congestion or load balancing effects). In both applications it is sensible to assume that the attacker has a budget~$B$ that is spent evenly among the resources with maximal load and that the leaders anticipate this strategy.
This motivates the definition of the following class of \emph{bilevel load balancing games} that, to the best of our knowledge, is new in the literature.

\begin{example}[Bilevel load balancing games]\label{ex:bilevel-cg}
Bilevel load balancing games are singleton resource graph games, where for every resource~$r$, the cost is of the form
\begin{equation}\label{eq:cost_attack}
c_r(\vec x):= x_r+\kappa_r^*(\vec x),
\text{ where } 
\kappa_r^*(\vec x)=\begin{cases} \frac{B}{|\arg\max_{r \in R} \{x_r \} |}, & \text{ if } r \in \arg\max_{r \in R} \{x_r \}, \\ 0, &\text{ else.}\end{cases}.\end{equation}
\end{example}

Finally, we mention that the interdependence of costs of resources on the loads of other resources has a long history in non-atomic traffic models. Dafermos~\cite{Dafermos71} proposes the use of such models to model the dependencies of the travel times on opposing directions of a two-lane road and on road segments leading to a common crossing. We obtain the following class of congestion games with non-separable costs as the natural atomic counterpart of the non-atomic traffic models with non-separable costs considered in the traffic literature (Dafermos \cite{Dafermos71,Dafermos72}, Smith~\cite{Smith79}).

\begin{example}[Unweighted network congestion games with non-separable costs]
These games are resource graph games, where the set of resources $R$ corresponds to the set of edges of a road network. For every player~$i$, the strategy set $X_i$ corresponds to the (indicator vectors of the edge set) of all paths between a source node $o_i$ and a destination node $d_i$ in the road network. One typically assumes that $\vec c : \R_{\geq 0}^R \to \R^m$ is monotonically non-decreasing, i.e., $c_r(\vec x) \leq c_r(\vec y)$ for all $x, y \in X$ with $x_s \leq y_s$ for all $s \in R$. 
\end{example}

In the examples discussed above, a mixed Nash equilibrium is guaranteed to exist due to Nash's theorem \cite{Nash50a}. However, as discussed, e.g., in Jiang and Leyton-Brown~\cite{JiangL-B07}, pure Nash equilibria are more favorable as a solution concept as they are easier to implement in practice. We are interested in identifying maximal conditions on the cost functions that ensure the existence of pure Nash equilibria in resource graph games. 

\subsection{Our results}

In this paper, we study the existence of pure Nash equilibria for resource graph games
with respect to the non-separable cost structures. We call a non-empty set $\C$ of cost functions \emph{consistent}, if every resource graph game with cost functions from $\C$ admits pure Nash equilibria.
We only require a natural condition on $\C$, namely that $\C$ is \emph{closed under composition}.
This means that for any two functions $\vec c_1,\vec c_2\in \C$ ($\vec c_1=\vec c_2$ is allowed)
acting on resource sets $R_1$ and $R_2$ with $|R_1|=m_1$ and $|R_2|=m_2$, respectively, the cost function $\vec c_1 \oplus \vec c_2 : \R^{m_1+m_2}_{\geq 0} \to \R^{m_1+m_2}$ defined as $\vec c_1 \oplus \vec c_2=(\vec c_1,\vec c_2)$ also belongs to $\C$.
This property naturally arises by composing two disjoint subsets $R_1$ and $R_2$ so that the cost structure within each set of the disjoint union is given by $\vec c_1$ and $\vec c_2$, and there is no interaction between the loads and costs of resources contained in the two different sets. We obtain the following results.
\begin{enumerate}
\item As our main result we show in Theorem~\ref{thm:consistency_unweighted} that a composition-closed set  $\C$ of cost functions 
is consistent if and only if for each $\vec c \in \C$ with $\vec c : \Z_{\geq 0}^m \to \R^m$ for some $m \in \N$, there are arbitrary functions $f_1,\dots,f_m : \Z_{\geq 0} \to \R$ and a symmetric matrix $\vec A \in \R^{m \times m}$ such that 
\begin{align*}
  \vec c(\vec x) = \bigl(f_1(x_1), \dots, f_m(x_m)\bigl)^\top + \vec A \vec x.
\end{align*}
Our result implies in particular that every resource graph game with this cost structure has a pure Nash equilibrium. This generalizes a result of Leyton-Brown and Tennenholtz~\cite{Leyton-BrownT03} who show that for the special case of local effect games with this cost structure, a pure Nash equilibrium exists. Our characterization also implies that for every other cost function $\tilde{\vec c}$ that does not adhere to this form, there is an unweighted resource graph game with costs defined by $\tilde{\vec c}$ that does not have a pure Nash equilibrium. For the proof of this result, we construct several highly symmetric resource graph games that allow to derive functional equations on the set of consistent cost functions that combined leave cost functions of the form above as the only possibility. Our results are also relevant for related work on the complexity of deciding the existence of pure Nash equilibria in action graph games (Jiang and Leyton-Brown~\cite{JiangL-B07}) as it implies that this computational problem is trivial for cost functions of the required form above.

\item We then study \emph{weighted} resource graph games, a natural generalization of resource graph games, where every player~$i$ has an intrinsic weight~$w_i$ and their strategy set is $X_i = \{w_i\, \vec x_i : \vec x_i \in Y_i\}$, where $Y_i \subseteq \{0,1\}^m$ is arbitrary. These games are relevant as a more fine-grained model for congestion games with non-separable costs, where the players have a different impact on the costs of the resources. We also provide a full characterization of the cost functions that are consistent for weighted resource graph games. Specifically, we show in Theorem~\ref{thm:consistency_weighted} that a composition-closed set $\C$ of continuous cost functions 
is consistent if and only if for each $\vec c \in \C$ with $\vec c : \R_{\geq 0}^m \to \R^m$ for some $m \in \N$, 
either $\vec c$ consists of separable exponential functions with a common exponent $\phi\in \R$, or 
there is a symmetric matrix $\vec A \in \R^{m \times m}$ and vector $\vec b\in \R^m$ such that 
$  \vec c(\vec x) = \vec A \vec x+\vec b.$
\item If the players' strategy spaces are restricted by some combinatorial property, then other non-separable
cost functions are possibly consistent.
In this regard, we consider matroid bases as such combinatorial domain. We show in Theorem~\ref{polymatroid:main} that pure Nash equilibria exist under a local monotonicity property, even when cost functions are player-specific. We demonstrate the applicability of this result by deriving an existence
 result of pure Nash equilibria for bilevel load balancing games as introduced in Example~\ref{ex:bilevel-cg}.  This classs of games is motivated by the study of network infrastructures facing external attackers and internal congestion effects.
\item Finally, in Section~\ref{app:complexity}, we discuss the computational complexity of deciding whether a given strategy profile is a pure Nash equilibrium and derive hardness results for network routing games and matroid games, respectively.
\end{enumerate}
\subsection{Related work}

Rosenthal~\cite{Rosenthal73a} shows that every unweighted congestion game with separable costs has a pure Nash equilibrium. Milchtaich~\cite{Milchtaich96} proposes two generalizations of unweighted congestion games. In the first generalization, called weighted congestion games, each player has a weight and the cost of each resource depends on the aggregated weight of its users. In the second generalization, called congestion games with player-specific costs, every player has an individual cost function for each resource.
Both generalizations alone still admit a pure Nash equilibrium for singleton congestion games, but the combination of both games fails to provide pure Nash equilibria, even for singleton games. The positive result for singletons is generalized by Ackermann et al.~\cite{Ackermann09} to games, where the strategy set of each player corresponds to the set of basis of a matroid. Weighted congestion games with general strategy spaces may fail to have a pure Nash equilibrium (Goemans et al.~\cite{Goemans05}, Libman and Orda~\cite{Libman01}), but have a pure Nash equilibrium for affine costs or exponential costs (Harks and Klimm~\cite{HarksK12}, Harks et al.~\cite{HarksKM11}, Fotakis et al.~\cite{Fotakis05}, Panagopoulou and Spirakis~\cite{Panagopoulou06}). Local effect games are introduced by Leyton-Brown and Tennenholtz~\cite{Leyton-BrownT03} who show that a pure Nash equilibrium exists when the mutual influence of different resources on the cost is linear and symmetric. Dunkel and Schulz~\cite{Dunkel08} show that for these games the computation of a pure Nash equilibrium is $\mathsf{PLS}$-complete. They also show that for  both local effect games with non-linear mutual effects and weighted congestion games with arbitrary cost functions, it is $\mathsf{NP}$-hard to decide whether a pure Nash equilibrium exists. The $\mathsf{PLS}$-completeness of computing a pure Nash equilibrium in unweighted congestion games with affine costs due to Ackermann et al.~\cite{Ackermann08} carries over to the weighted case. 

Action graph games are introduced by Bhat and Leyton-Brown~\cite{BhatL04} and Jiang et al.~\cite{JiangLB11} as a generalization of local-effect games. They show that every strategic game can be represented as an action graph game. Daskalakis et al.~\cite{DaskalakisSVV09} give a fully polynomial-time approximation scheme (FPTAS) for computing an approximate mixed equilibrium in action graph games with constant degree, constant treewidth, and a constant number of agent types. They also give several hardness results for the case that one of the conditions on the game is violated. Jiang and Leyton-Brown~\cite{JiangL-B07} show that for symmetric action graph games played on a graph of bounded treewidth, it can be decided efficiently whether a pure Nash equilibrium exists while the problem is $\mathsf{NP}$-hard to decide in general.

Resource graph games are introduced as by Jiang et al.~\cite{JiangCL17} as a further generalization of action graph games. Chan and Jiang~\cite{ChanJ18} give an FPTAS for computing an approximate Nash equilibrium in resource graph games with a constant number of player types and further restrictions on the strategy sets.

Congestion games with non-atomic players where the load of one resource has an impact on the cost of another resource are usually called \emph{congestion games with non-separable costs}. They were first proposed by Dafermos~\cite{Dafermos71,Dafermos72}. She shows that the equilibrium condition can be formulated as an optimization problem, if the Jacobian of the cost function is symmetric. Smith~\cite{Smith79} provides a variational inequality for the non-symmetric case. Perakis~\cite{Perakis07} studies the price of anarchy of non-atomic congestion games with linear non-separable costs of the form $c(\vec x) = \vec A \vec x + \vec b$.

  Bilevel, Stackelberg, or Leader-Follower games, have been studied extensively over the last years. In these games, the players are partitioned into \emph{leaders}, acting first, and \emph{followers}, choosing their strategy only after the leaders' choices become apparent. Such hierarchical relationships appear in many real-world problems, e.g. in pricing or toll setting problems~\cite{LabbeV16,CorreaGLNS18, HarksSV19}, 
security games~\cite{SinhaFAKT18,JiangPQST13}, fare evasion games~\cite{CHKM17},
  supply chain and marketing management~\cite{Xiuli2007}, or in voting scenarios~\cite{Xia2010}.  
In the context of bilevel games with congestion effects,
Castiglioni et al.~\cite{CastiglioniMGC19} and Marchesi et al.~\cite{MarchesiC019} considered Stackelberg games with an underlying unweighted congestion game. However, they assume that there is only one leader and the leader participates in the same congestion game as the followers (but the leader's congestion cost functions may be different from the followers'). Depending on the structure of strategy spaces and congestion cost functions, they analyze the computational complexity of computing Stackelberg equilibria. In particular, they devise efficient algorithms for singleton strategy spaces, where either all followers have the same strategies~\cite{CastiglioniMGC19}, or the followers can be divided in ``classes'' having the same strategies~\cite{MarchesiC019}. 
The case of multiple leaders playing an unweighted or weighted congestion game
subject to followers affecting the resource costs (as for instance by attacks as modeled in Example~\ref{ex:bilevel-cg}) is, to the best of our knowledge, completely open.

\section{Preliminaries}
For an integer $k \in \N$, let $[k] := \{1,\dots,k\}$. Let $N = [n]$ be a finite set of players and $R = [m]$ be a finite set of $m$ resources. For each player~$i$, the set of strategies available to player~$i$ is an arbitrary set $X_i \subseteq \{0,1\}^m$.
We call $x = (\vec x_1,\dots, \vec x_n)$ with $x_i \in X_i$ for all $i \in N$ a strategy profile and $X = X_1 \times \dots \times X_n$ the strategy space. We use standard game theory notation; for a strategy profile $x \in X$, we write $x = (\vec x_i, x_{-i})$ meaning that $\vec x_i$ is the strategy that player~$i$ plays in $x$ and $x_{-i}$ is the partial strategy profile of all players except $i$. Every strategy profile $x = (\vec x_1,\dots,\vec x_n) \in X$ induces a load vector $\vec x = \sum_{i \in [n]} \vec x_i \in \R_{\geq 0}^m$.
For a set $S \subseteq R = [m]$, we denote by $\mathbf{1}_{S}$ the indicator vector of set $S$ in $\R^m$.
We are further given a cost function $\vec c : \R^m_{\geq 0} \to \R^m$. Usually, one assumes that $\vec c$ has a succinct representation of the following form. For every resource $r \in R$, there is a neighborhood $B_r \subseteq R$ such that $c_r(\vec x)$ is independent of $x_s$ for all $s \notin B_r$, i.e., $c_r(\vec x) = c_r(\vec y)$ for all $\vec x, \vec y \in X$ with $x_s = y_s$ for all $s \in B_r$. If this is the case, and $|B_r| \leq k$ for all $r \in R$ the function $\vec c$ can be encoded by $mn^k$ numbers since it suffices to specify for each $r \in R$ the value of $c_r(\vec x)$ as a function of the $n^k$ possible load vectors of the resources in $B_r$. Intuitively, the function $\vec c$ maps a load vector $\vec x \in \R^m_{\geq 0}$ to a cost vector $\vec c(\vec x) \in \R^m$, i.e., $\vec c(\vec x) = (c_1(\vec x),\dots,c_m(\vec x))^\top$ and for a resource $r \in R$ the cost experienced by players using $e$ when the congestion vector is $\vec x$ is $c_r(\vec x)$.
The strategic game $G = (N, X, (\pi_i)_{i \in N})$ where the private cost of player~$i$ in strategy profile $x \in X$ is defined as $\pi_i(x) = \vec x_i^\top c(\vec x) = \sum_{r \in R} x_{i,r} c_r(\vec x)$ is called a \emph{resource graph game}. 

We also consider a generalization of resource graph game to weighted players. In a \emph{weighted resource graph game}, every player~$i \in N$ has a weight $w_i \in \R_{>0}$. The strategy set of player~$i$ is then defined as $X_i = \{w_i\, \vec x_i : \vec x_i \in Y_i\}$ where $Y_i \subseteq \{0,1\}^m$ is arbitrary. Compared to unweighted resource graph games, in a weighted resource graph game, the set of vectors in the strategy of player~$i$ is multiplied with the scalar $w_i$. The weighted resource graph game is then the strategic game $G = (N, X, (\pi_i)_{i \in N})$ where $\pi_i$ is defined as before. 

A strategy profile $x \in X$ is a \emph{pure Nash equilibrium}, if $\pi_i(x) \leq \pi_i(\vec y_i, x_{-i})$ for all $i \in N$ and $\vec y_i \in X_i$.
For a non-empty set $\C$ set of cost functions,
we are interested in establishing conditions on $\C$ that ensure that every
resource graph game with cost functions from $\C$ admits pure Nash equilibria.
We require a mild technical assumption on $\C$, namely that $\C$ is \emph{closed under composition} in the following sense. 
Let $\vec c, \vec c' \in \C$ with $\vec c : \R_{\geq 0}^m \to \R^m$ and $\vec c' : \R_{\geq 0}^{m'} \to \R^{m'}$. Then, we require that the function $\vec c \oplus \vec c' : \R^{m+m'}_{\geq 0} \to \R^{m+m'}_{\geq 0}$ defined as
$
\vec c \oplus \vec c'
(
\vec x,
\vec y) =
(
\vec c(\vec x),
\vec c'(\vec y)
)
$
is also contained in $\C$. This is an intuitive property of a set of functions $\C$ for the following reasons. The cost functions $\vec c$, $\vec c'$ each define a cost structure on sets of resources $R$, $R'$ with $|R|=m$ and $|R'| = m'$. The cost function $\vec c \oplus \vec c' : \R^{m+m'}_{\geq 0} \to \R^{m+m'}$ then defines a cost structure on the disjoint union of $R$ and $R'$,  where the cost structure within each set of the disjoint union is given by $\vec c$ or $\vec c'$, and there is no interaction between the loads and costs of resources contained in the two different sets. In particular, for any $k \in \N$ and any $\vec c \in \C$, $\vec c : \R_{\geq 0}^m \to \R^m$, the $k$-fold disjoint union $\vec c \oplus \dots \oplus \vec c : \R^{km}_{\geq 0} \to \R^{km}$ is contained in $\C$. In the following, we denote the $k$-fold disjoint union of $\vec c$ by $\vec c^k$.
For a set $\C$ of cost functions as above, we say that $\C$ is \emph{consistent for unweighted resource graph games}, if for every $\vec c \in \C$, we have that every unweighted resource graph game with costs given by $\vec c$ has a pure Nash equilibrium. Recall that when $\vec c : \R_{\geq 0}^m \to \R^m$ with $m \in \N$, then every unweighted resource graph game with costs given by $\vec c$ has $m$ resources. Consistency for weighted resource graph games is defined analogously.

\section{Resource graph games with unweighted players}
\label{sec:unweighted}

In this section, we consider unweighted resource graph games. Since in such a game,
the load on each resource is a nonnegative integer, it is without loss of generality to assume that 
the domain of all cost functions is the non-negative integer lattice, that is, they are of the form $\vec c : \Z^m_{\geq 0} \to \R^m$ for some $m \in \N$.
Our main result gives a complete characterization of consistency for unweighted resource graph games.

\begin{theorem}
\label{thm:consistency_unweighted}
Let $\C$ be a set cost functions that is closed under composition. Then the following two statements are equivalent:
\begin{enumerate}
    \item\label{item:main1} $\C$ is consistent for unweighted resource graph games.
    \item\label{item:main2} For each $\vec c \in \C$ with $\vec c : \Z_{\geq 0}^m \to \R^m$ for some $m \in \N$, there are functions $f_1,\dots,f_m : \Z_{\geq 0} \to \R$ and a symmetric matrix $\vec A \in \R^{m \times m}$ such that 
\begin{align}
  \vec c(\vec x) = \bigl(f_1(x_1), \dots, f_m(x_m)\bigl)^\top + \vec A \vec x.\label{eq:cost-function-unweighted}
\end{align}
\end{enumerate}
In particular, the set $\C^*$ of all cost functions of the form \eqref{eq:cost-function-unweighted} is the unique maximal set of cost functions that is closed under composition and consistent for unweighted resource graph games.
\end{theorem} 
We subdivided the proof of both directions in the following subsections.
\subsection{Proof of Theorem~\ref{thm:consistency_unweighted}: \ref{item:main2}. $\Rightarrow$ \ref{item:main1}.}

We first prove that statement $\ref{item:main2}.$ of Theorem~\ref{thm:consistency_unweighted} implies consistency of $\C$. Observe that any composition of functions of form \eqref{eq:cost-function-unweighted} is again of form \eqref{eq:cost-function-unweighted}. It is thus sufficient to show existence of a pure Nash equilibrium for any any unweighted resource graph game with a cost function of this form. 

\begin{restatable}{lemma}{restateLemUnweightedSufficient}
\label{lem:consistency_unweighted_sufficient}
Let $G$ be an unweighted resource graph game on $m$ resources with cost function $\vec c : \R_{\geq 0}^m \to \R^m$ given by $\vec c(\vec x) = \bigl(f_1(x_1), \dots, f_m(x_m)\bigl)^\top + \vec A \vec x$, where $f_1,\dots,f_m : \Z_{\geq 0} \to \R$ are arbitrary functions and $\vec A \in \R^{m \times m}$ is a symmetric matrix. Then $G$ has a pure Nash equilibrium.
\end{restatable}

\begin{proof}
Fix an arbitrary unweighted resource graph game $G$ whose cost is determined by $\vec c \in \C$ and an arbitrary strategy profile $x \in X$.
Consider the process of adding the players to the game in order $1,\dots,n$ and let us sum their private costs. In the following, we write $\vec x_{\leq i} = \sum_{j \in N : j \leq i} \vec x_i$ for the load vector of players up to $i$. Let $P(x)$ be the sum of the private costs of the player added to the game when adding them in order $1,\dots,n$. We define the function $\vec f : \Z_{\geq 0}^m \to \R^m$ as $\vec f(\vec x) = \bigl(f_1(x_1),\dots,f_m(x_m)\bigr)^\top$ and obtain
\begin{align*}
P(x) &= \sum_{i \in N} \vec x_i^\top \Bigl[ \vec f(\vec x_{\leq i}) + \vec A \vec x_{\leq i}	\Bigr].
\end{align*}
We have
\begin{align*}
\sum_{i \in N} \vec x_i^\top \vec f(\vec x_{\leq i}) = \sum_{r \in R} \sum_{k = 1}^{x_r} f_r(k)
\end{align*}
as well as
\begin{align*}
\sum_{i \in N} \vec x_i^\top \vec A \vec x_{\leq i}	&= \sum_{i \in N} \vec x_i^{\top} \vec A \Bigl(\sum_{j \in N : j \leq i} \vec x_j\Bigl)\\
& = \frac{1}{2} \sum_{i \in N} \sum_{j \in N} \vec x_i^\top \vec A \vec x_j + \frac{1}{2} \sum_{i \in N} \vec x_i^\top \vec A \vec x_i \\
&= \frac{1}{2}\vec x^\top \vec A \vec x + \frac{1}{2} \sum_{i \in N} \vec x_i^\top \vec A \vec x_i,
\end{align*}
where for the second equation we used the symmetry of $\vec A$. We obtain
\begin{align*}
P(x) = \sum_{r \in R} \sum_{k=1}^{x_r} f_r(k)+ \frac{1}{2} \vec x^\top \vec A \vec x + \frac{1}{2} \sum_{i \in N} \vec x_i^\top  \vec A \vec x_i.
\end{align*}
This shows that $P(x)$ is invariant under a reordering of the players. Next, consider a deviation of an arbitrary player. Since $P(x)$ is invariant under a reordering of the players, it is without loss of generality to assume that player~$n$ deviates.
We obtain
\begin{align*}
P(\vec y_i, x_{-i}) - P(x) &= \vec y_n^\top \Bigl[ \vec f(\vec x_{\leq n-1} + \vec y_n) + \vec A(\vec x_{\leq n-1} + y_n)\Big]\\
&\quad - \vec x_n^\top \Bigl[ \vec f(\vec x_{\leq n}) + \vec A(\vec x_{\leq n})\Big] \\
&= \pi_n(\vec y_i, x_{-i}) - \pi_n(x).
\end{align*}
We conclude that $P$ is an exact potential function and, hence,  $G$ admits a pure Nash equilibrium. Since $G$ was chosen arbitrarily, the result follows.
\end{proof}

\subsection{Proof of Theorem~\ref{thm:consistency_unweighted}: \ref{item:main1}. $\Rightarrow$ \ref{item:main2}.}

In the following, we show that statement 2 of Theorem~\ref{thm:consistency_unweighted} is a necessary condition for the consistency of $\C$.
We prove this by constructing for any given $\vec c \in \C$ a family of different resource graph games whose cost functions are $4$-fold compositions of $\vec c$. All these games will have the following symmetry property, which we will use to establish that $\vec c$ is indeed of the form \eqref{eq:cost-function-unweighted}. 

\begin{definition}
Let $A, B \in \mathbb{R}$. We say a game $G = (N, X, (\pi_i)_{i \in N})$ is $(A, B)$-symmetric for players $i, j \in N$, if for any strategy profile $x \in X$, the following two statements are fulfilled:
\begin{itemize}
  \item $\pi_i(x) = A$ and $\pi_j(x) = B$ or $\pi_i(x) = B$ and $\pi_j(x) = A$.
  \item There are $\vec y_i \in X_i$ and $\vec y_j \in X_j$ such that $\pi_i(\vec y_i, x_{-i}) = \pi_j(x)$ and $\pi_j(\vec y_j, x_{-j}) = \pi_i(x)$.
\end{itemize}
\end{definition}

The following lemma shows a key property for $(A,B)$-symmetric games.

\begin{lemma}
\label{lem:symmetric-game}
  If a game $G$ is $(A, B)$-symmetric for players $i, j \in N$ and admits a pure Nash equilibrium, then $A = B$.
\end{lemma}

\begin{proof}
  Let $x \in X$ be a pure Nash equilibrium for $G$.
  Because $G$ is symmetric for $i$ and $j$, there are are $\vec y_i \in X_i$ and $\vec y_j \in X_j$ such that $\pi_i(\vec y_i, x_{-i}) = \pi_j(x)$ and $\pi_j(\vec y_j, x_{-j}) = \pi_i(x)$. Because $x$ is a pure Nash equilibrium, we obtain
  \begin{align*}
    \pi_i(x) \leq \pi_i(\vec y_i, x_{-i}) = \pi_j(x) \leq \pi_j(\vec y_j, x_{-j}) = \pi_i(x)
  \end{align*}
  and hence $\pi_i(x) = \pi_j(x)$. Note that symmetry of $G$ implies $\{A, B\} = \{\pi_i(x), \pi_j(x)\}$ and therefore $A = B$.
\end{proof}

We proceed to prove a first functional equation that needs to be satisfied for a set of consistent cost functions that is closed under composition. The equation states that a discrete version of the Jacobian of the cost function must be symmetric. For the proof, we construct a suitable $(A,B)$-symmetric game.

\begin{lemma}
\label{lem:symmetric-jacobi-unweighted}
Let $\C$ be closed under composition and consistent for unweighted resource graph games. Then, for all $\vec c \in \C$,  $\vec c : \Z_{\geq 0}^m \to \R^m$ with $m \in \N$, we have 
  $$c_r(\vec x + \vec 1_{\{r, s\}}) - c_r(\vec x + \vec 1_{r}) = c_s(\vec x + \vec 1_{\{r, s\}}) - c_s(\vec x + \vec 1_{s})$$ 
for all $r, s \in [m]$ and all $\vec x \in \mathbb{Z}_{\geq 0}^m$.
\end{lemma}

\begin{figure}
  \centering
  
  \scalebox{0.9}{%
  \begin{tikzpicture}[scale=1.8]
      
    \begin{scope}
      \draw[rounded corners, fill=orange!20, opacity=0.5] (0.7, -0.3) -- ++(0, 0.9) -- ++(3, 0) -- 
        ++(0, 0.7) -- ++(0.6, 0) -- ++(0, -0.9) -- ++(-3, 0) -- ++(0, -0.7) -- cycle;
    \end{scope}
      
    \begin{scope}[xshift=1.5cm, yshift=0.5cm]
      \draw[fill=blue!20, opacity=0.5, rotate=-45] (0, 0) ellipse (1cm and 0.3cm);
    \end{scope}
      
    \begin{scope}[xshift=2.5cm, yshift=0.5cm]
      \draw[fill=red!20, opacity=0.5, rotate=-45] (0, 0) ellipse (1cm and 0.3cm);
    \end{scope}
      
    \begin{scope}[xshift=3.5cm, yshift=0.5cm]
      \draw[fill=cyan!20, opacity=0.5, rotate=-45] (0, 0) ellipse (1cm and 0.3cm);
    \end{scope}
    
    \node at (1.35, 0.8) {$\vec x_{1}$};
    \node at (3.35, 0.8) {$\vec x'_{1}$};
    
    \node at (2.35, 0.8) {$\vec x_{2}$};
    \node at (4.1, 0.6) {$\vec x'_{2}$};
    
    \foreach \x in {1,2,3,4}
    {
    
        \node[res] (r\x) at (\x,1) {$r_{\x}$};
        \node[res] (s\x) at (\x,0) {$s_{\x}$};
        
    }
    
    \draw[black!40, <->, thick] (r1) -- (s1);
    \draw[black!40, <->, thick] (r2) -- (s2);
    \draw[black!40, <->, thick] (r3) -- (s3);
    \draw[black!40, <->, thick] (r4) edge[bend right=10] (s4);

    \node at (2.5, -0.5) {(i)};
  \end{tikzpicture}
  \hspace{0.8cm}
  \begin{tikzpicture}[scale=1.8]
      
    \begin{scope}
      \draw[rounded corners, fill=blue!20, opacity=0.5] (0.7, -0.3) -- ++(0, 1.6) -- ++(1.6, 0) -- 
        ++(0, -0.6) -- ++(-1, 0) -- ++(0, -1) -- cycle;
    \end{scope}
      
    \begin{scope}[xscale=-1, xshift=-5cm]
      \draw[rounded corners, fill=cyan!20, opacity=0.5] (0.7, -0.3) -- ++(0, 1.6) -- ++(1.6, 0) -- 
        ++(0, -0.6) -- ++(-1, 0) -- ++(0, -1) -- cycle;
    \end{scope}
      
    \begin{scope}
      \draw[rounded corners, fill=red!20, opacity=0.5] (1.8, -0.3) -- ++(0, 1.9) -- ++(2.4, 0) -- 
        ++(0, -0.9) -- ++(-0.4, 0) -- ++(0, 0.7) -- ++(-1.6, 0) -- ++(0, -1.7) -- cycle;
    \end{scope}
      
    \begin{scope}[xscale=-1, xshift=-5cm]
      \draw[rounded corners, fill=orange!20, opacity=0.5] (1.8, -0.3) -- ++(0, 2.2) -- ++(2.4, 0) -- 
        ++(0, -1.2) -- ++(-0.4, 0) -- ++(0, 1) -- ++(-1.6, 0) -- ++(0, -2) -- cycle;
    \end{scope}
    
    \node at (1.5, 1) {$\vec x_{1}$};
    \node at (3.5, 1) {$\vec x'_{1}$};
    
    \node at (4, 1.45) {$\vec x_{2}$};
    \node at (1, 1.5) {$\vec x'_{2}$};
    
    \foreach \x in {1,2,3,4}
    {
        \node[res] (r\x) at (\x,1) {$r_{\x}$};
        \node[res] (s\x) at (\x,0) {$s_{\x}$};
        \draw[black!40, <->, thick] (r\x) -- (s\x);
    }
    
    \node at (2.5, -0.5) {(ii)};
        
  \end{tikzpicture}
  }
  \vspace{-0.3cm}
  \caption{Games constructed for the proofs of Lemmas~\ref{lem:symmetric-jacobi-unweighted} and \ref{lem:game2}, respectively. Each clique represents a copy of the resource set (resources other than $r, s$ and dummy players are omitted). Player $1$ chooses among strategies $\vec x_1$ and $\vec x'_1$, player $2$ chooses among strategies $\vec x_2$ and $\vec x'_2$.\label{fig:games1-2}}
\end{figure}
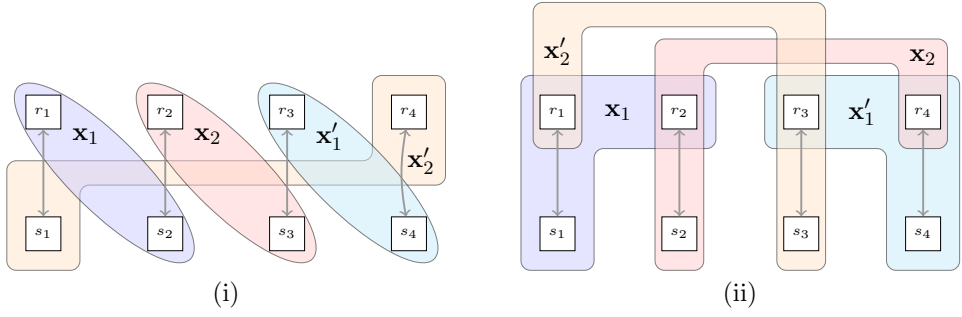

\begin{proof}
Let $\vec c \in \C$, $r, s \in [m]$ with $r \neq s$, let $\vec x \in \Z^m_{\geq 0}$ be arbitrary. 
Since $\C$ is closed under composition, we have that $\vec c^4 : \Z_{\geq 0}^{4m} \to \R^{4m}$ is also contained in $\C$.
Consider the following game with $4m$ resources and cost function $\vec c^{4}$. 
For $k \in [4]$ and $t \in [m]$ we denote the $k$-th copy of resource $t$ by $t_k$.
For each original resource $t \in [m]$, there are $x_{t}$ dummy players whose only strategy is $\sum_{k \in [4]} \mathbf{1}_{t_k}$. In addition there are two players $1$ and $2$ with strategy sets $X_{1} = \{\mathbf{1}_{\{r_1, s_2\}},\, \mathbf{1}_{\{s_3, r_4\}}\}$ and $X_{2} = \{\mathbf{1}_{\{s_1, r_3\}},\, \mathbf{1}_{\{r_2, s_4\}}\}$. See Figure~\ref{fig:games1-2}(i) for a depiction of the strategy space.

 It is straightforward to check that the above game is $(A, B)$-symmetric for players $1$ and $2$ with
 \begin{align*}
 A = c_r(\vec x + \vec 1_{\{r,s\}}) + c_s(\vec x + \vec 1_{s}) \quad \text{and} \quad B = c_r(\vec x + \vec 1_{r}) + c_s(\vec x + \vec 1_{\{r, s\}}).
 \end{align*}
Since $\C$ is consistent, the thus defined game has a pure Nash equilibrium and we conclude
  $A = B$ by Lemma~\ref{lem:symmetric-game}, which completes the proof of the lemma.
\end{proof}

The following two lemmas establish that the discrete Hessian of each $c_r$ for $r \in [m]$ must be diagonal. For the proof of these two lemmas, we use the symmetry of the Jacobian shown in Lemma~\ref{lem:symmetric-jacobi-unweighted} together with suitably constructed $(A,B)$-symmetric games.

\begin{lemma}
\label{lem:game2}
Let $\C$ be closed under composition and consistent for unweighted resource graph games. Then, for all $\vec c \in \C$, $\vec c : \Z_{\geq 0}^m \to \R^m$ with $m \in \N$, the following two functional equations are satisfied for all $r, s \in [m]$ with $r \neq s$ and all $\vec x \in \mathbb{Z}_{\geq 0}^R$ with $x_r > 0$:
  \begin{itemize}
   \item[(a)] $c_r(\vec x + \vec 1_{s}) - c_r(\vec x) = c_r(\vec x + \vec 1_{\{r, s\}}) - c_r(\vec x + \vec 1_{r})$ and
   \item[(b)] $c_r(\vec x + 2 \cdot \vec 1_{s}) - c_r(\vec x + \vec 1_{s}) = c_r(\vec x + \vec 1_{s}) - c_r(\vec x)$.
 \end{itemize}
\end{lemma}

\begin{proof}
Let $\vec c \in \C$, $\vec c : \Z_{\geq 0}^m \to \R^m$, $r, s \in [m]$ with $r \neq s$, and $\vec x \in \Z^m_{\geq 0}$ with $x_r > 0$ be arbitrary. Since $\C$ is closed under composition, the function $\vec c^4 : \Z_{\geq 0}^{4m} \to \R^{4m}$ is also contained in $\C$. Consider the following game with $4m$ resources and cost function $\vec c^{4}$. For $k \in [4]$ and $t \in [m]$, we denote the $k$-th copy of resource $t$ by $t_k$.
For each original resource $t \in [m] \setminus \{r\}$, there are $x_{t}$ dummy players whose only strategy is $\sum_{k \in [4]} \mathbf{1}_{t_k}$. There also are $x_r - 1$ dummy players for resource $r$ whose only strategy is $\sum_{k \in [4]} \mathbf{1}_{r_k}$.
  In addition there are two players $1$ and $2$ with strategy sets 
  \begin{align*}
     X_{1} = \big\{\mathbf{1}_{\{r_1, s_1, r_2\}},\; \mathbf{1}_{\{r_3, r_4, s_4\}}\big\} \quad \text{and} \quad
     X_{2} = \big\{\mathbf{1}_{\{r_1, r_3, s_3\}},\; \mathbf{1}_{\{r_2, s_2, r_4\}}\big\}.
   \end{align*}
  See Figure~\ref{fig:games1-2}(ii) for a depiction of the strategy space.
  It is easy to check that the above game is $(A, B)$-symmetric for players $1$ and $2$ with 
  \begin{align*}
    A & = c_r(\vec x + \vec 1_{\{r,s\}}) \,+\, c_s(\vec x + \vec 1_{\{r, s\}}) \,+\, c_r(\vec x)  \text{ and }\\
    B & = c_r(\vec x + \vec 1_{s}) \,+\, c_s(\vec x + \vec 1_{s}) \,+\, c_r(\vec x + \vec 1_{\{r, s\}}).
  \end{align*}
  By consistency of $\vec c$, the game must have a pure Nash Equilibrium and thus $A = B$ by Lemma~\ref{lem:symmetric-game}. Subtracting the first and third term of $A$ and the second term of $B$ on both sides yields
  \begin{align}
    c_s(\vec x + \vec 1_{\{r, s\}}) - c_s(\vec x + \vec 1_{s}) & = c_r(\vec x + \vec 1_{s}) - c_r(\vec x).\label{eq:game-2}
  \end{align}
  Applying Lemma~\ref{lem:symmetric-jacobi-unweighted} to the left-hand side of \eqref{eq:game-2} yields
  \begin{align*}
     c_r(\vec x + \vec 1_{\{r, s\}}) - c_r(\vec x + \vec 1_{r}) = c_r(\vec x + \vec 1_{s}) - c_r(\vec x),
  \end{align*}
  which proves (a).
  
  Applying Lemma~\ref{lem:symmetric-jacobi-unweighted} to the right-hand side of \eqref{eq:game-2} instead, yields
  \begin{align*}
     c_s(\vec x + \vec 1_{\{r, s\}}) \,-\, c_s(\vec x + \vec 1_{s}) & = c_s(\vec x + \vec 1_{s}) - c_s(\vec x - \vec{1}_r + \vec{1}_s),
  \end{align*}
  which is equivalent to (b) when substituting $\vec x$ for $\vec x - \vec 1_r + \vec 1_s$ and then swapping the roles of $r$ and~$s$.
\end{proof}

\begin{lemma}
\label{lem:game3}
Let $\C$ be closed under composition and consistent for unweighted resource graph games. Then, for all $\vec c \in \C$, $\vec c : \Z_{\geq 0}^m \to \R^m$ with $m \in \N$ we have
  $$c_r(\vec x + \vec 1_{s}) \,-\, c_r(\vec x) \;=\; c_r(\vec x + \vec 1_{\{s, t\}}) \,-\, c_r(\vec x + \vec 1_{t})$$ 
  for all $r, s, t \in [m]$ with $r, s, t$ pairwise distinct and all $\vec x \in \mathbb{Z}_{\geq 0}^m$ with $x_r > 0$.
\end{lemma}

\begin{proof}
  Let $\vec c \in \C$, $\vec c : \Z_{\geq 0}^m \to \R^m$, $m \in \N$ be arbitrary, let $r, s, t \in [m]$ be pairwise distinct, and let $\vec x \in \Z^m_{\geq 0}$ with $x_r > 0$. 
  Let $\vec x' = \vec x_r - \vec 1_r$.
  Consider the following game with $4m$ resources and cost function $\vec c^{4}$. For $k \in [4]$ and $u \in [m]$, we denote the $k$-th copy of resource $u$ by $u_k$.
  For each resource $u \in [m] \setminus \{r\}$, there are $x'_{u}$ dummy players whose only strategy is $\sum_{k \in [4]} \mathbf{1}_{u_k}$.
  In addition there are two players $1$ and $2$ with strategy sets 
  \begin{align*}
     X_{1} = \big\{\mathbf{1}_{\{r_1, s_2, t_2\}},\,  \mathbf{1}_{\{s_3, t_3, r_4\}}\big\} \quad \text{and} \quad
     X_{2} = \big\{\mathbf{1}_{\{s_1, t_1, r_4\}},\,  \mathbf{1}_{\{r_2, s_4, t_4\}}\big\}.
   \end{align*}
   See Figure~\ref{fig:game3} for a depiction of the strategy space.
  It is easy to check that the above game is $(A, B)$-symmetric for players $1$ and $2$ with 
  \begin{align*}
    A & \; = \; c_r(\vec x' + \vec 1_{\{r,s,t\}}) \,+\, c_s(\vec x' + \vec 1_{\{s, t\}}) \,+\, c_t(\vec x' + \vec 1_{\{s, t\}}), \quad  \text{ and }\\
    B & \; = \; c_r(\vec x' + \vec 1_{r}) \,+\, c_s(\vec x' + \vec 1_{\{r,s,t\}}) \,+\, c_t(\vec x' + \vec 1_{\{r, s, t\}}).
  \end{align*}
  By consistency of $\mathcal{C}$ the game must have a pure Nash Equilibrium and thus $A = B$ by Lemma~\ref{lem:symmetric-game}. By subtracting the second term of $A$ and the first and third term of $B$ from both sides we obtain
  \begin{align*}
    & c_r(\vec x' + \vec 1_{\{r,s,t\}}) \,-\, 
    \underbrace{\big(c_t(\vec x' + \vec 1_{\{r, s, t\}}) - c_t(\vec x' + \vec 1_{\{s, t\}}) \big)}_{=\; c_r(\vec x' + \vec 1_{\{r, s, t\}}) - c_r(\vec x' + \vec 1_{\{r, s\}})} \,-\, c_r(\vec x' + \vec 1_{r}) \\ 
    = \ & 
    \underbrace{c_s(\vec x' + \vec 1_{\{r,s,t\}}) \,-\, c_s(\vec x' + \vec 1_{\{s, t\}})}_{=\; c_r(\vec x' + \vec 1_{\{r,s,t\}}) - c_r(\vec x' + \vec 1_{\{r, t\}})}.
  \end{align*}
  Applying the identities indicated above, which follow from Lemma~\ref{lem:symmetric-jacobi-unweighted}, and then using $\vec x' = \vec x - \vec 1_r$ yields $c_r(\vec x + \vec 1_{s}) - c_r(\vec x) = c_r(\vec x + \vec 1_{\{s, t\}}) - c_r(\vec x + \vec 1_{t})$.
\end{proof}

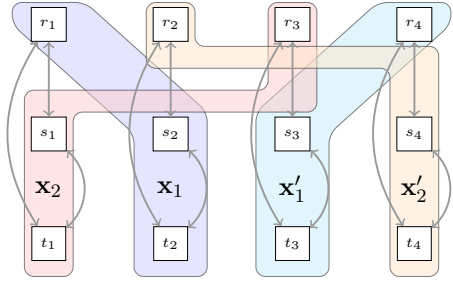
\begin{figure}
  \centering
  
  \scalebox{0.9}{%
  \hspace{0.2cm}
  \begin{tikzpicture}[scale=1.8, yscale=0.9]
      
    \begin{scope}
      \draw[rounded corners, fill=blue!20, opacity=0.5] (1.7, -0.3) -- ++(0, 1.3) -- ++(-1, 1) -- 
        ++(0, 0.2) -- ++(0.6, 0) -- ++(1, -1) -- ++(0, -1.5) -- cycle;
    \end{scope}
    
    \begin{scope}[xscale=-1, xshift=-5cm]
      \draw[rounded corners, fill=cyan!20, opacity=0.5] (1.7, -0.3) -- ++(0, 1.3) -- ++(-1, 1) -- 
        ++(0, 0.2) -- ++(0.6, 0) -- ++(1, -1) -- ++(0, -1.5) -- cycle;
    \end{scope}

    \begin{scope}
      \draw[rounded corners, fill=red!20, opacity=0.5] (0.8, -0.3) -- ++(0, 1.7) -- ++(2, 0) -- 
        ++(0, 0.8) -- ++(0.4, 0) -- ++(0, -1) -- ++(-2, 0) -- ++(0, -1.5) -- cycle;
    \end{scope}
      
    \begin{scope}[xscale=-1, xshift=-5cm]
      \draw[rounded corners, fill=orange!20, opacity=0.5] (0.8, -0.3) -- ++(0, 2.1) -- ++(2, 0) -- 
        ++(0, 0.4) -- ++(0.4, 0) -- ++(0, -0.6) -- ++(-2, 0) -- ++(0, -1.9) -- cycle;
    \end{scope}
      
    \node at (1, 0.5) {$\vec x_{2}$};
    \node at (4, 0.5) {$\vec x'_{2}$};
    
    \node at (2, 0.5) {$\vec x_{1}$};
    \node at (3, 0.5) {$\vec x'_{1}$};

    \foreach \x in {1,2,3,4}
    {
    
        \node[res] (r\x) at (\x,2) {$r_{\x}$};
        \node[res] (s\x) at (\x,1) {$s_{\x}$};
        \node[res] (t\x) at (\x,0) {$t_{\x}$};

        \draw[black!40, thick, <->] (r\x) -- (s\x);
        \draw[black!40, thick, <->] (s\x) edge[bend left=45] (t\x);
        \draw[black!40, thick, <->] (t\x) edge[bend left] (r\x);

    }
        
  \end{tikzpicture}
  }
  
  \caption{Game constructed for the proof of Lemma~\ref{lem:game3}. Each clique represents a copy of the resource set (resources other than $r, s, t$ and dummy players are omitted). Player $1$ chooses among strategies $\vec x_1$ and $\vec x'_1$, player $2$ chooses among strategies $\vec x_2$ and $\vec x'_2$.\label{fig:game3}}
\end{figure}

Given the form of the discrete Hessian established in Lemmas~\ref{lem:game2}~and~\ref{lem:game3}, we conclude now that the influence of the load on a resource $s$ on the cost of some other resource~$r$ must be linear. 
This is formalized in the following lemma, which follows by inductively applying our previous results.

\begin{restatable}{lemma}{restateLemDiagonalHesse}
\label{lem:diagonal-hesse-unweighted}
Let $\C$ be closed under composition and consistent for unweighted resource graph games. Then, for all $\vec c \in \C$, $\vec c : \Z_{\geq 0}^m \to \R^m$ with $m \in \N$ we have
  $$c_r(\vec x + \vec 1_{s}) - c_r(\vec x) = c_r(\vec y + \vec 1_{s}) - c_r(\vec y)$$ 
  for all $r, s \in [m]$ with $r \neq s$ and all $\vec x, \vec y \in \mathbb{Z}_{\geq 0}^m$ with $x_r, y_r > 0$.
\end{restatable}

\begin{proof}
Let $\vec c \in \C$, $\vec c : \Z_{\geq 0}^m \to \R^m$, $r, s \in [m]$ with $r \neq s$, and $\vec x, y \in \Z^m_{\geq 0}$ with $x_r, y_r > 0$ be arbitrary.
  We show the lemma by induction on $k = \sum_{r' \in R} |x_{r'} - y_{r'}|$. If $k = 0$, then $x = y$ and the claim is trivially fulfilled. Thus assume $k > 0$. Without loss of generality, there is $t \in R$ with $x_{t} > y_{t}$.
  Let $\vec x' := \vec x - \vec 1_{t}$. Note that $x'_r \geq y_r > 0$, so that Lemma~\ref{lem:game2} and Lemma~\ref{lem:game3} can be applied to $\vec x'$. We distinguish three cases for $t$.
  
\paragraph{Case $t = r$:} In this case, $\vec x = \vec x' + \vec 1_r$. We obtain
  \begin{align*}
    c_r(\vec x + \vec 1_{s}) - c_r(\vec x) \; = \; c_r(\vec x' + \vec 1_{\{r,s\}}) - c_r(\vec x' + \vec 1_{r}) \; = \;  c_r(\vec x' + \vec 1_{s}) - c_r(\vec x').
  \end{align*}
  where the second equality follows from Lemma~\ref{lem:game2}.
 
\paragraph{Case $t = s$:} In this case, $\vec x = \vec x' + \vec 1_s$. We obtain
  \begin{align*}
    c_r(\vec x + \vec 1_{s}) - c_r(\vec x) \; = \; c_r(\vec x' + 2 \cdot \vec 1_{s}) - c_r(\vec x' + \vec 1_{s}) \; = \;  c_r(\vec x' + \vec 1_{s}) - c_r(\vec x').
  \end{align*}
  where the second equality follows from Lemma~\ref{lem:game2}.
 
\paragraph{Case $t \in R \setminus \{r, s\}$:} In this case, $\vec x = \vec x' + \vec 1_t$. We obtain
  \begin{align*}
    c_r(\vec x + \vec 1_{s}) - c_r(\vec x) \; = \; c_r(\vec x' + \vec 1_{\{s, t\}}) - c_r(\vec x' + \vec 1_{t}) \; = \;  c_r(\vec x' + \vec 1_{s}) - c_r(\vec x').
  \end{align*}
  where the second equality follows from Lemma~\ref{lem:game3}.
  
  In either case, $c_r(\vec x + \vec 1_{s}) - c_r(\vec x) = c_r(\vec x' + \vec 1_{s}) - c_r(\vec x')$, which is equal to $c_r(\vec y + \vec 1_{s}) - c_r(\vec y)$ by the induction hypothesis because  $\sum_{r' \in R} |x'_{r'} - y_{r'}| < k$.
\end{proof}

Combining Lemmas~\ref{lem:symmetric-jacobi-unweighted} and~\ref{lem:diagonal-hesse-unweighted}, we observe that the interaction effects of distinct resources in the cost function $\vec c$ must be linear and symmetric and thus condition \eqref{eq:cost-function-unweighted} is indeed necessary for consistency. This is formalized in the following lemma, which completes the proof of Theorem~\ref{thm:consistency_unweighted}.

\begin{restatable}{lemma}{restateLemNecessaryUnweighted}
\label{lem:necessary_unweighted}
Let $\C$ be closed under composition and consistent for unweighted resource graph games. Then, for all $\vec c \in \C$, $\vec c : \Z_{\geq 0}^m \to \R^m$ with $m \in \N$ there are $m$ functions $f_1, \dots, f_m : \Z_{\geq 0} \to \R$ and a symmetric matrix $\vec A \in \R^{m \times m}$ such that
$\vec c(\vec x) = (f_1(x_1), \dots, f_m(x_m))^\top + \vec A_{r,\cdot} \, \vec x$
for all $\vec x \in \Z^R_{\geq 0}$.
\end{restatable}

\begin{proof}
  For $r \in R$ define $f_r : \R_{\geq 0} \to \R$ by $f_r(x) = c_r(x \cdot \vec 1_r)$ for all $x \in \R_{\geq 0}$.
  The matrix $\vec A \in \R^{m \times m} = (a_{r,s})_{r,s \in [m]}$ is defined as
\begin{align*}
a_{r,s} =  
\begin{cases}
c_r(\vec 1_{\{r, s\}}) - c_r(\vec 1_{r}) & \text{ if $r\neq s$ },\\
0 & \text{ otherwise }.
\end{cases}
\end{align*}
Note that $a_{r,s} = a_{s,r}$ and, hence $\vec A$ is symmetric by Lemma~\ref{lem:symmetric-jacobi-unweighted}.
  
Let $\vec x \in \R_{\geq 0}^m$ and $r \in [m]$ be arbitrary. 
For $r \in [m]$, we denote by $\vec A_{r,\cdot}$ its $r$-th row. 
We show that $c_r(\vec x) = f_r(x_r) + A_{r,\cdot} \vec x$ by induction on $k = \sum_{s \in R \setminus \{r\}} x_s$. For $k = 0$ the claim is true by definition of $f_r$. Thus assume $k > 0$ and let $s \in R \setminus \{r\}$ with $x_s > 0$. We obtain
  \begin{align*}
    c_r(\vec x) & = c_r(\vec x - \vec 1_s) + c_r(\vec 1_{\{r, s\}}) - c_r(\vec 1_{r})\\
    & = f_r(x_r) + \vec A_{r,\cdot} (\vec x - \vec 1_s) + a_{r,s}\\
    & = f_r(x_r) + \vec A_{r,\cdot} \vec x
  \end{align*}
  where the first identity follows from Lemma~\ref{lem:diagonal-hesse-unweighted}
  and the second identity follows from the induction hypothesis and the definition of $a_{r,s}$.
\end{proof}

\section{Resource graph games with weighted players}

In this section, we establish necessary and sufficient conditions for consistency when each player $i \in N$ imposes a weight $w_i \in \R_{\geq 0}$ on the resources in their strategy. The characterization reveals two possible cases: A consistent set of cost functions either contains only affine functions with a symmetric Jacobian, or the cost functions of individual resources are exponential and separable (i.e., there is no interaction among distinct resources).

\begin{theorem}
\label{thm:consistency_weighted}
Let $\C$ be a set of continuous cost functions that is closed under composition. Then $\C$ is consistent for weighted resource graph games if and only if one of the following two statements is fulfilled:
\begin{enumerate}
  \item\label{enum:weighted1} For each $\vec c \in \C$ with $\vec c : \R_{\geq 0}^m \to \R^m$ for some $m \in \N$ there is a symmetric matrix $\vec A \in \R^{m \times m}$ and a vector $\vec b \in \R^m$ such that 
  $\vec c(\vec x) = \vec A \vec x + \vec b$.
  \item\label{enum:weighted2} There is $\phi \in \R$ such that for all $\vec c \in \C$ with $\vec c : \R_{\geq 0}^{m} \to \R^m$ there are $\vec a, \vec b \in \R^m$ such that $c_r(\vec x) = a_r \exp(\phi x_r) + b_r$ for all $r \in [m]$ and all $\vec x \in \R^m_{\geq 0}$.
\end{enumerate}
\end{theorem} 

The two distinct cases arise due to the fact that \emph{weighted congestion games} are a special case of weighted resource graph games, namely where the cost of each resource~$r$ depends on the load of $r$ only, i.e., $B_r = \{r\}$ for all $r \in R$.
For these games Harks and Klimm~\cite{HarksK12} provided a characterization that shows that consistent sets contain only affine or only exponential cost functions. 
In the following, we give prove the sufficiency and necessity of either of the two conditions.

\subsection{Proof of Theorem~\ref{thm:consistency_weighted}: \ref{enum:weighted1}. or \ref{enum:weighted2}. $\Rightarrow$ Consistency of $\C$:}

We show sufficiency of conditions~1 or 2 in Theorem~\ref{thm:consistency_weighted}, respectively, for the consistency of $\C$.
If condition~2\ is fulfilled, then any weighted resource graph game with cost function $\vec c \in \C$ is a weighted congestion game with exponential costs. For these games, the existence of pure Nash equilibria has been established in \cite{HarksK12}, Theorem~5.1. It is therefore sufficient to show that condition 1 of Theorem~\ref{thm:consistency_weighted} is also sufficient for consistency. The following lemma establishes the sufficiency of condition 1, following the same lines as the proof of Lemma~\ref{lem:consistency_unweighted_sufficient}.

\begin{restatable}{lemma}{restateLemWeightedSufficient}
\label{lem:consistency_weighted-sufficient}
Let $G$ be a weighted resource graph game on $m$ resources with cost function $\vec c : \R_{\geq 0}^m \to \R^m$ given by $\vec c(\vec x) = \vec A \vec x + \vec b$, where $\vec A \in \R^{m \times m}$ is a symmetric matrix and $\vec b \in \R^m$ is  a vector. Then $G$ has a pure Nash equilibrium.
\end{restatable}

\begin{proof}
Fix an arbitrary weighted resource graph game $G$ whose cost is determined by $\vec c \in \C$ and an arbitrary strategy profile $x \in X$. As in the proof of Theorem~\ref{thm:consistency_unweighted}, let $P(x)$ be the sum of the private costs of the players when adding them to the game in order $1,\dots,n$. We again write $\vec x_{\leq i} = \sum_{j \in N : j \leq i} \vec x_i$ for the load vector of the players up to $i$. 
We then obtain
\begin{align*}
P(x) = \sum_{i \in N} \vec x_i^\top \Bigl[ \vec A \vec x_{\leq i} + \vec b\Bigr].
\end{align*}
Similarly to the proof of Theorem~\ref{thm:consistency_unweighted}, we calculate
\begin{align*}
\sum_{i \in N} \vec x_i^\top \vec A \vec x_{\leq i} &= \sum_{i \in N} \vec x^\top_i \vec A \Bigl(\sum_{j \in N : j \leq i} \vec x_j\Bigr) \\
&= \frac{1}{2} \sum_{i \in N} \sum_{j \in N} \vec x_i^\top \vec A \vec x_j + \frac{1}{2} \sum_{i \in N} \vec x_i^\top \vec A \vec x_i \\
&= \frac{1}{2} \vec x^\top \vec A \vec x + \frac{1}{2} \sum_{i \in N} \vec x^\top_i \vec A \vec x_i, 	
\end{align*}
as in the unweighted case where we again used the symmetry of $\vec A$. We obtain
\begin{align*}
P(x) &=  \frac{1}{2} \vec x^\top \vec A \vec x + \frac{1}{2} \sum_{i \in N} \vec x_i^\top \Bigl[\vec A \vec x_i + \vec b\Bigr].
\end{align*}
This shows that $P$ is independent of the ordering of the player and the remainder of the proof is equivalent to Theorem~\ref{thm:consistency_unweighted}.
\end{proof}

\subsection{Proof of Theorem~\ref{thm:consistency_weighted}: Consistency of $\C$ $\Rightarrow$ \ref{enum:weighted1}.  or \ref{enum:weighted2}.:}
We now prove the necessity of the conditions given in Theorem~\ref{thm:consistency_weighted}.
By a slight adaptation of the constructions in Section~\ref{sec:unweighted}, we obtain the following stronger version of Lemmas~\ref{lem:symmetric-jacobi-unweighted} and~\ref{lem:diagonal-hesse-unweighted} for cost functions that are consistent for weighted players.

\begin{restatable}{lemma}{restateLemStructureWeighted}
\label{lem:structure-weighted-players}
Let $\C$ be closed under composition and consistent for weighted resource graph games. Then, for all $\vec c \in \C$, $\vec c : \R_{\geq 0}^m \to \R^m$ with $m \in \N$ the following functional equations are satisfied:
  \begin{itemize}  
    \item[(a)] $c_r(\vec x + \varepsilon \cdot \vec 1_{\{r,s\}}) - c_r(\vec x + \varepsilon \cdot \vec 1_{r}) = c_s(\vec x + \varepsilon \cdot \vec 1_{\{r,s\}}) - c_f(\vec x + \varepsilon \cdot \vec 1_{s})$ for all $r, s \in R$ and all $\vec x \in \R_{\geq 0}^m$ and all $\varepsilon > 0$ and
    \item[(b)] $c_r(\vec x + \varepsilon \cdot \vec 1_{s}) - c_r(\vec x) = c_r(\vec y + \varepsilon \cdot \vec 1_{s}) - c_r(\vec y)$ for all $r, s \in R$ with $r \neq s$ and all $\vec x, \vec y \in \R_{\geq 0}^m$ with $x_r, y_r > 0$ and all $\varepsilon > 0$.
  \end{itemize}
\end{restatable}

\begin{proof}[Proof (Sketch).]
  We follow the same constructions used to establish Lemmas~\ref{lem:symmetric-jacobi-unweighted}, \ref{lem:game2}, \ref{lem:game3}, and \ref{lem:diagonal-hesse-unweighted}. However, we set $w_1 = w_2 = \varepsilon$ and adjust the weights of the dummy players for each resource such that the load on the resource equals the corresponding coordinate of $\vec x$.
\end{proof}

The following lemma follows from the characterization of consistent functions for weighted congestion games with separable cost functions due to Harks and Klimm~\cite{HarksK12}.

\begin{restatable}{lemma}{restateLemWeightedCongestion}
\label{lem:weighted-congestion-separable}
Let $\C$ be a set of continuous functions that is closed under composition and consistent for weighted resource graph games. Then,  for all $\vec c \in \C$, $\vec c : \R_{\geq 0}^m \to \R^m$ with $m \in \N$ one of the following statements is true:
  \begin{itemize}
    \item[(a)] For all $S \subseteq R$ and all $\vec z \in \R^m_{\geq 0}$ there are $a_{S, \vec z}, b_{S, \vec z} \in \R$ such that $\sum_{r \in S} c_r(\vec z + \lambda \vec 1_{S}) = a_{S, \vec z} \lambda + b_{S, \vec z}$ for all $\lambda \geq 0$.
    \item[(b)] There is $\phi \in \R$ such that for all $S \subseteq R$ and all $\vec z \in \R^m_{\geq 0}$ there are $a_{S, \vec z}, b_{S, \vec z} \in \R$ such that $\sum_{r \in S} c_r(\vec z + \lambda \vec 1_{S}) = a_{S, \vec z} \exp(\phi \lambda) + b_{S, \vec z}$ for all $\lambda \geq 0$.
  \end{itemize}
\end{restatable}

\begin{proof}
Let $\vec c \in \C$, $\vec c : \R_{\geq 0}^m \to \R^m$ with $m \in \N$ be arbitrary.
For $S \subseteq R$ and $\vec z \in \R^m_{\geq 0}$ define $c_{S, \vec z} : \R_{\geq 0} \to \R$ by $c_{S, \vec z}(\lambda) = \sum_{r \in S} c_r(\vec z + \lambda \vec 1_{S})$. Let $\mathcal{C}' = \{c_{S, \vec z} : S \subseteq R, \vec z \in \R^m\}$ be the set of all functions arising in this way. We show that any weighted congestion game with separable cost functions on $k$ resources where each resource has a cost function $c' \in \mathcal{C}'$ is isomorphic to a weighted resource graph game with cost function $\vec c^{k}$ on $km$ resources. Since $\C$ is closed under composition, the function $\vec c^k$ is contained in $\C$, and, hence, consistency of $\C$ for weighted resource graph games implies the consistency of $\mathcal{C}'$ for weighted congestion games. It is known (\cite{HarksK12}, Theorem 5.1) that a set of continuous functions is consistent for weighted congestion games if and only if it contains only affine functions (as described in case (a) of the lemma) or it only contains only exponential functions (as described in case (b) of the lemma). Hence the lemma follows from the following construction.
    
  Consider any weighted congestion game $G'$ with arbitrary player set $N'$, weights $w'_i$ for each $i \in N'$, strategies $X'_i = \{w_i \cdot \vec x_i : \vec x_i \in Y_i\}$ with $Y_i \subseteq \{0,1\}^k$, and resource set $R' = \{r'_1,\dots,r'_k\}$ such that for all $r \in R'$ the cost function $c'_r : \R_{\geq 0} \to \R$ of resource $r$ is of the form $c'_r = c_{S, \vec z}$ for some $S \subseteq R$ and $\vec z \in \R^R_{\geq 0}$. 
In what follows, we construct an isomorphic weighted resource graph game $G$ with player set $N$, $mk$ resources, and cost function $\vec c^k$.
For $j \in [k]$ let $S_j \subseteq R$ and $\vec  z_j = (z_{j,1,},\dots,z_{j,m}) \in \R^m_{\geq 0}$ be such that $c'_{r_j} = c_{S_j, \vec z_j}$.
We define $N = N' \cup \{(r, j) : r \in R, j \in [k]\}$, i.e., the set of players $N$ of $G$ contains the player set $N'$ of the original congestion game plus $mk$ additional dummy players.
Each dummy player $(r, j)$ can only play strategy $\mathbf{1}_{r_j}$ where $r_j$ is the $j$-th copy of resource $r \in [m]$. That dummy player has a weight $w_{(r, j)} = z_{j,r}$.
Each normal player $i \in N'$ has the same weight $w_i = w'_i$ as in the original congestion game.
For each strategy $\vec x'_i \in X'_i$ of player~$i$ in the original congestion game, there is a strategy $\vec x_i \in X_i$ that arises from $\vec x'_i$ by replacing each resource $r_j \in R'$ by the set of resources $S_j \subseteq R$, i.e., $\vec X_i = \{ \sum_{j \in [k]} x_{i,j}' \, \mathbf{1}_{S_j} : \vec x_i' = (x'_{i,1},\dots,x'_{i,k}) \in X_i' \}$.
Thus, there is a one-to-one correspondence between strategy profiles $\vec x'$ for $G'$ and strategy profiles $\vec x$ of $G$ and it is easy to see that by construction, the private cost of player $i \in N'$ is the same for $\vec x'$ in $G'$ and the corresponding profile $\vec x$ in $G$.
\end{proof}

Equipped with Lemmas~\ref{lem:structure-weighted-players} and ~\ref{lem:weighted-congestion-separable}, we can show that the impact of the load of resource~$s$ on the cost of resource~$r$ needs to be linear and symmetric. In addition, the impact is non-existent if case (a) of  Lemma~\ref{lem:weighted-congestion-separable} does not hold.
This is formalized in the following lemma.

\begin{lemma}
\label{lem:weighted-linear-effect}
Let $\C$ be a set of continuous functions that is closed under composition and consistent for weighted resource graph games.
Then, for all $\vec c \in \C$, $\vec c : \R_{\geq 0}^m \to \R^m$ with $m  \in \N$  and all $r, s \in [m]$ with $r \neq s$ there is $a_{r,s} = a_{s, r}$ such that
$$c_r(\vec z + \lambda \vec 1_{s}) - c_r(\vec z) = a_{r,s} \lambda$$ for all $\vec z \in \R^m_{\geq 0}$ with $z_r > 0$ and all $\lambda \geq 0$. Moreover, if case (a) of Lemma~\ref{lem:weighted-congestion-separable} does not hold, then $a_{r,s} = 0$ for all $r, s \in R$.
\end{lemma}

\begin{proof}
Let $\vec c \in \C$, $\vec c : \R_{\geq 0}^m \to \R^m$ with $m \in \N$ be arbitrary. 
  Let $r, s \in [m]$ with $r \neq s$ and let $\vec z \in \R_{\geq 0}^m$ with $z_r > 0$. 
  Using Lemma~\ref{lem:structure-weighted-players} we obtain
  \begin{align*}
    \underbrace{c_r(\vec z + \lambda \vec 1_{\{r, s\}}) + c_s(\vec z + \lambda \vec 1_{\{r, s\}})}_{h_0(\lambda)} & = c_r(\vec z + \lambda \vec 1_{r}) + c_r(\vec z + \lambda \vec 1_{\{r, s\}}) - c_r(\vec z + \lambda \vec 1_{r})\\[-15pt]
    & \qquad + c_s(\vec z + \lambda \vec 1_{f}) + c_f(\vec z + \lambda \vec 1_{\{r, s\}}) - c_s(\vec z + \lambda \vec 1_{s})\\
    & = \underbrace{c_r(\vec z + \lambda \vec 1_{r})}_{h_1(\lambda)} + \underbrace{c_s(\vec z + \lambda \vec 1_{s})}_{h_2(\lambda)} + 2\big(\underbrace{c_r(\vec z + \lambda \vec 1_{\{r, s\}}) - c_r(\vec z + \lambda \vec 1_{r})}_{h_4(\lambda)}\big)
  \end{align*}
  for all $\lambda \geq 0$.
  We apply Lemma~\ref{lem:weighted-congestion-separable} to the expressions $h_0(\lambda), h_1(\lambda), h_2(\lambda)$ and distinguish two cases.
  
  If we are in case (a) of Lemma~\ref{lem:weighted-congestion-separable}, then all three expressions are affine functions of $\lambda$ and we conclude that also $h_4$ must be affine in $\lambda$, i.e., there is $a, b \in R$ such that
  $c_r(\vec z + \lambda \vec 1_{\{r, s\}}) - c_r(\vec z + \lambda \vec 1_{r}) = a \lambda + b$.
  By part (a) of Lemma~\ref{lem:structure-weighted-players}, we observe that this equality also holds (for the same values of $a$ and $b$) when swapping the roles of $s$ and $r$.
  Applying part (b) of Lemma~\ref{lem:structure-weighted-players}, we observe that $b = 0$ and $a$ is independent of $z$, thus proving the statement of the lemma for this case.
  
   If we are in case (b) of Lemma~\ref{lem:weighted-congestion-separable}, then all three expressions are exponential functions of the form $a \exp(\phi \lambda) + b$ for some $\phi \in \R$ and we conclude that also $h_4$ must be of this form, i.e., there is $a', b' \in R$ such that $c_r(\vec z + \lambda \vec 1_{\{r, s\}}) - c_r(\vec z + \lambda \vec 1_{r}) = a' \exp(\phi \lambda) + b'$. Applying part (b) of Lemma~\ref{lem:structure-weighted-players}, we conclude that $a' = b' = 0$, thus proving the statement of the lemma for this case.
\end{proof}

We are now ready to establish the necessity of condition~1 or 2 of Theorem~\ref{thm:consistency_weighted}, concluding the proof of the theorem.

\begin{lemma}\label{lem:necessary_weighted}
Let $\C$ be a set of continuous functions that is closed under composition and consistent for weighted resource graph games. 
Then, one of the following statements is true:
\begin{enumerate}
  \item For each $\vec c \in \C$ with $\vec c : \R_{\geq 0}^m \to \R^m$ for some $m \in \N$ there is a symmetric matrix $\vec A \in \R^{m \times m}$ and a vector $\vec b \in \R^m$ such that 
  $\vec c(\vec x) = \vec A \vec x + \vec b$.
  \item There is $\phi \in \R$ such that for all $\vec c \in \C$ with $\vec c : \R_{\geq 0}^{m} \to \R^m$ there are $\vec a, \vec b \in \R^m$ such that $c_r(\vec x) = a_r \exp(\phi x_r) + b_r$ for all $r \in [m]$ and all $\vec x \in \R^m_{\geq 0}$.
\end{enumerate}
\end{lemma}

\begin{proof}
Let $\vec c \in \C$, $\vec c : \R_{\geq 0}^m \to \R^m$ with $m \in \N$ 
be arbitrary. By Lemma~\ref{lem:weighted-linear-effect}, for all $r,s \in [m]$, there is $a_{r,s}$ such that $c_r(\vec z + \lambda \mathbf{1}_s) - c_r(\vec z) = a_{r,s} \lambda$ for all $\vec z \in \R_{\geq 0}^m$ with $z_r > 0$ and all $\lambda \geq 0$.
  Let $r_1, \dots, r_m$ be an arbitrary ordering of the resources in $R$ with $r_m = r$.
  Defining $\vec x^{(0)} = \vec x$ and $\vec x^{(i)} = \vec x^{(i - 1)} - x_{r_i} \cdot \vec 1_{r_i}$ for $i \in [m]$ we obtain
  \begin{align*}
    c_r(\vec x) & = c_r(x_r \cdot \vec 1_{r}) + \sum_{i = 1}^{m-1} c_r(\vec x^{(i - 1)}) - c_r(\vec x^{(i)})\\
    &= c_r(x_r \cdot \vec 1_{r}) + \sum_{s \in R \setminus \{r\}} a_{r,s} (x_s - 1).
  \end{align*}
  In case (a) of Lemma~\ref{lem:weighted-congestion-separable}, we conclude that $c_r(x_r \cdot \vec 1_{r})$ is an affine function of $x_r$. This implies that there is a symmetric matrix $\vec A \in \R^{m \times m}$ and a vector $\vec b \in \R^m$ such that $c_r(\vec x) = \vec A_{r,\cdot} \vec x + b_r$ for all $r \in R$ and all $\vec x \in \R^m_{\geq 0}$ with $x_r > 0$. Since $\C$ is closed under composition, all functions $\vec c \in \C$ have this property thus we retrieve case 1 of Lemma~\ref{lem:necessary_weighted}.
  
  In case (b) of Lemma~\ref{lem:weighted-congestion-separable}, we conclude that $c_r(x_r \cdot \vec 1_{r})$ is an exponential function of $x_r$. By Lemma~\ref{lem:weighted-linear-effect}, we than have that $a_{r,s} = 0$ and we thus obtain that $c_r(\vec x) = a_r \exp(\phi x_r) + b_r$ for all $r \in [m]$ for some constant $a_r,b_r, \phi \in \R$. As $\C$ is closed under composition, this implies that all functions $\vec c \in \C$ have this property, and we retrieve case 2 of Lemma~\ref{lem:necessary_weighted}.
\end{proof}

\section{Resource graph games on matroids}
While our previous characterizations hold for \emph{arbitrary strategy
spaces} of the players, we now turn to \emph{restricted strategy spaces}.
Specifically, we consider \emph{matroidal} strategy spaces, where the strategy set $X_i$ of player $i\in N$
corresponds to the set of incidence vectors of bases of a player-specific matroid $M_i=(R,\B_i), i\in N$ defined on the resource set $R$.

A \emph{resource graph game on matroids} is  then represented by the tuple 
$ G=(N,X,C),$
where $C= (\vec c_i)_{i\in N}$ denotes the vector of \emph{player-specific
non-separable cost functions}. 
The private cost of
player~$i$ under strategy profile $x \in X$ is defined as  $\pi_i(x) = \vec x_i^\top \vec c_i(\vec x).$
In order to specify assumptions on the functions $\vec c_i(\vec x), i\in N$,
we introduce the concept of player \emph{types}.
Let $\mathcal{B}$ denote the set of feasible
matroid base systems represented in binary vectors over $\{0,1\}^m$.
For example $U_k \in \mathcal{B}$, where $U_k$ is the base set of the uniform matroid of rank $ 1\leq k\leq m$.
For a matroid game $ G=(N,X,C)$, we say that  player $i\in N$ is of type $T \in \mathcal{B}$,
if the matroid base system for $i$ is given by $X_i=T$.
Now we
define a general notion of local monotonicity
of non-separable functions.
\begin{definition}[Local Monotonicity]\label{ass:local}
A function $\vec c: \Z_{\geq 0}^{m}\rightarrow \R^m$ 
is \emph{locally monotone}, if
for all $T\in \mathcal{B},r\in R$, there are non-decreasing functions
 $\nu_{T,r}:\Z_{\geq 0}\rightarrow\R, $ such that for all $\vec t \in T$, all $\vec z \in \Z^m_+,$  and all $r,s\in R$ with 
$\vec u=\vec t+\vec 1_s-\vec 1_r\in T$ and $\nu_{T,r}(t_r)\leq \nu_{T,f}(u_s)$
it holds that
\begin{align}\label{eq:cost}
 \vec t^\top \vec c(\vec t+\vec z)  &\leq \vec u^\top \vec c(\vec u+\vec z).
\end{align}
\end{definition}
The locality aspect of Definition~\ref{ass:local} arises, because
the condition relates  the cost of playing
strategies $\vec t$ and $\vec u$, which differ
only by exchanging the entries of two elements.
While this definition is quite abstract, we will provide an application and illustrating example in the realm of bilevel load balancing games, formalized in Theorem~\ref{thm:load-balancing} and its proof given below.
The main idea for our existence proof is to  use an associated matroid congestion game with \emph{separable} player-specific non-decreasing
cost functions $\nu(\vec x):=(\bm\nu_i(\vec x))_{i \in N}$ in order to construct a Nash equilibrium
for the original game with non-separable functions.
The proof of the theorem is given in Section~\ref{polymatroid:main-proof}.
\begin{theorem}\label{polymatroid:main}
Let $ G=(N,X,C)$ be a resource graph game on matroids with
$C= (\vec c_i)_{i\in N}$ and
locally monotone cost functions $\vec c_i,i\in N$.
Then, the following statements hold.
\begin{enumerate}
\item\label{enum:1} Any profile  $x \in X$ that is a Nash equilibrium for the matroid congestion game $ G=(N,X, \nu)$ is also a Nash equilibrium for the resource graph game $G=(N,X,C)$.
\item\label{enum:2} Nash equilibria for the resource graph game $G=(N,X,C)$ do exist.
\end{enumerate}
\end{theorem}

As a consequence of Theorem~\ref{polymatroid:main}, we obtain the following application to the class of bilevel load balancing games on matroids. The proof of the theorem is given in Section~\ref{thm:load-balancing-proof}.
  \begin{theorem}\label{thm:load-balancing}
Bilevel load-balancing games on matroids possess pure Nash equilibria.
  \end{theorem}

\subsection{Proof of Theorem~\ref{polymatroid:main}}
\label{polymatroid:main-proof}
We recap a trivial property of an equilibrium for the matroid game with separable cost functions.
\begin{lemma}
Let  $x \in X$ be an equilibrium for the matroid congestion game $ G=(N,X,\nu),$
with player-specific separable and non-decreasing cost functions $\nu(\vec x):=(\nu_i(\vec x), i\in N)$. Then, for any $y \in X$ with
$\vec y_i=\vec x_i+\vec 1_s-\vec 1_r$ for some $i\in N$ and $\vec y_j=\vec x_j$ for all $j \neq N \setminus \{i\}$, we have
$\nu_{i,r}(x_r)\leq \nu_{i,s}(y_s)$.
\end{lemma}

\begin{proof}
Assume by contradiction $\nu_{i,r}(x_r)> \nu_{i,s}(y_s)$. Then, by the monotonicity
and separability of $\bm\nu_i$, we have
$\pi_i(y)<\pi_i(x)$, contradiction.
\end{proof}
Now we prove Theorem~\ref{polymatroid:main}.
\begin{proof}
For~\ref{enum:1}.:\\
Let $x \in X$ be an equilibrium for the matroid congestion game $ G=(N,X,\nu)$.
Now we evaluate the cost of a player $i\in N$ when switching from $\vec x_i \in X_i$
to some $\vec y_i \in X_i$ for the game $G=(N,X,C)$:
\begin{align*}
    \pi_i(x)&=\vec x_i^\top \vec c_i(\vec x)
=\sum_{r\in \text{supp}(\vec x_i)} c_{i,r}(\vec x)
=\sum_{r\in \text{supp}(\vec x_i)\setminus \text{supp}(\vec y_i)} c_{i,r}(\vec x)
+\sum_{r\in \text{supp}(\vec x_i)\cap \text{supp}(\vec y_i)} c_{i,r}(\vec x).
\end{align*}
Because $X_i$ consists of the bases of a matroid, 
the profile $\vec y_i$ can be decomposed into
a sequence of single-element exchanges of the form
\[ \vec y_i =\vec x_i+\sum_{j=1}^k (\vec 1_{s_j}-\vec 1_{r_j})\]
with $r_1, \dots, r_k, s_1, \dots, s_k \in R$, $r_j\neq r_{j'}$ and $s_j\neq s_{j'}$ for all $ j \neq j'$, and such that $\vec y_i^{\ell}= \vec x_i+\sum_{j=1}^\ell (\vec 1_{s_j}-\vec 1_{r_j})\in X_i$ for all $1\leq \ell\leq k$.
We denote by $y^\ell:=(\vec y_i^{\ell},x_{-i})$ the corresponding profile, where
only player $i$ changed the strategy according to $\vec y_i^{\ell}$.
We prove by induction over $1\leq \ell\leq k-1$ 
that $\pi_i(y^{\ell})\leq\pi_i(y^{\ell+1})$.
For $\ell=1$, we get
\[ \pi_i(x)= \vec x_i^\top \vec c_i(\vec x)\leq {\vec y^1_i}^\top \vec c_i(\vec y^1)=\pi_i(y^1),
\]
where the  inequality follows by $\nu_{i,s_1}(y^{1}_{s_1})\geq \nu_{i,r_1}(x_{r_1})$ and~\eqref{eq:cost}. Here, we used the local monotonicity property of Definition~\ref{ass:local} by identifying $\vec x_{-i}= \vec u, \vec x_i=\vec t$ and $ \vec y^1_i=\vec y$.
Now we consider the inductive step $\ell\rightarrow \ell+1$:
\begin{align*}
\pi_i(y^{\ell})= {\vec y^{\ell}_i} ^\top \vec c_i(\vec y^{\ell}) \leq {\vec y^{\ell+1}_i} ^\top \vec c_i(\vec y^{\ell+1})=\pi_i(y^{\ell+1}),
\end{align*}
which again follows by $\nu_{i,s_{\ell+1}}( y^{\ell}_{s_{\ell+1}})\geq \nu_{i,r_{\ell+1}}(y_{r_{\ell+1}})=\nu_{i,r_{\ell+1}}(x_{r_{\ell+1}})$ and~\eqref{eq:cost}.

For~\ref{enum:2}.: We use ~\ref{enum:1}. together with the fact that pure Nash equilibria
do exist for the matroid congestion game $ G=(N,X,\nu)$
with non-decreasing separable player-specific cost functions, see Ackermann, R\"oglin and V\"ocking~\cite{Ackermann09}.
\end{proof}

\subsection{Proof of Theorem~\ref{thm:load-balancing}}\label{thm:load-balancing-proof}
In the following lemma, we verify that bilevel load balancing games as introduced in Example~\ref{ex:bilevel-cg} satisfy the conditions of Definition~\ref{ass:local}
and thus possesses Nash equilibria. Instead of only considering singleton strategies as in Example~\ref{ex:bilevel-cg}, we allow bases of matroids that may be player-specific.
\begin{lemma}
The cost function defined in~\eqref{eq:cost_attack} 
is locally monotone.
\end{lemma}
\begin{proof}
We need to show that there are functions $\nu(\vec x):=(\bm\nu_T(\vec x))_{T \in \B}$ such that
\eqref{eq:cost} holds under the  assumptions stated in Definition~\ref{ass:local}.
Let us rewrite~\eqref{eq:cost} in the current setting:
\begin{equation}\label{eq:cost-attack}
\begin{aligned}
 \vec t^\top \vec c(\vec t+\vec z)&=\sum_{g\in  \text{supp}(\vec t)}
 t_g+z_g+\kappa_g^*(\vec t+\vec z)\\
 & \leq 
  \sum_{g\in  \text{supp}(\vec u)}
 u_g+z_g+\kappa_g^*(\vec u+\vec z)= \vec u^\top \vec c(\vec u+\vec z).
 \end{aligned}
\end{equation}

We claim that by setting $\nu_{T,g}(x)=\nu_{T,g}(x)=x$ for all $x \in \Z_+, T\in \B,g\in R$ inequality
\eqref{eq:cost-attack} holds. Let $\vec t \in T$, $\vec z \in \Z^m_+,$  and consider $r,s\in R$ with 
$\vec u=\vec t+\vec 1_s-\vec 1_r\in T$ and  $\nu_{T,r}(t_r+z_r)=t_r+z_r\leq \nu_{T,f}(u_s+z_s)=u_s+z_s$.
 By $t_r+z_r\leq u_s+z_s$ and $t_g+z_g=u_g+z_g$ for all $g\notin \{r,s\}$, we get
 \[ \sum_{g\in  \text{supp}(\vec t)}
 t_g+z_g \leq \sum_{g\in  \text{supp}(\vec u)} u_g+z_g.\]
 Thus, it suffices to show that
  \begin{equation}\label{eq:kappa} \sum_{g\in  \text{supp}(\vec t)}
\kappa_g^*(\vec t+\vec z) \leq \sum_{g\in  \text{supp}(\vec u)} \kappa_g^*(\vec u+\vec z).\end{equation}

For $\vec w\in \Z_+^R$ define $S(\vec w):=\arg\max\{w_g \vert g\in R\}$.
We distinguish  two cases.
 \begin{itemize}
  \item Case 1: $s\notin S(\vec u+\vec z)$.

  With $u_s+z_s\geq t_r+z_r$, we get  $r\notin S(\vec t+\vec z)$ and thus
$\sum_{g\in  \text{supp}(\vec t)} \kappa_g^*(\vec t+\vec z)=\sum_{g\in  \text{supp}(\vec u)} \kappa_g^*(\vec u+\vec z)$, hence~\eqref{eq:kappa} follows.\\
 
 \item Case 2: $s\in S(\vec u+\vec z)$.
  \end{itemize}
 We consider the two sub-cases of whether or not
 $s$ was already an argmax element under $\vec t+\vec z$ or not.
\begin{itemize}
\item[-]  Case 2(a): $s\in S(\vec t+\vec z)$.
 
 This case implies $S(\vec u+\vec z)=\{s\}$ and thus we get 
\[ \sum_{g\in  \text{supp}(\vec u)} \kappa_g^*(\vec u+\vec z)= \kappa_s^*(\vec u+\vec z)=B\geq 
\sum_{g\in  \text{supp}(\vec t)} \kappa_g^*(\vec t+\vec z).\]
\item[-]  Case 2(b): $s\notin S(\vec t+\vec z)$. 

This case implies
 $S(\vec u+\vec z)=S(\vec t+\vec z)\cup\{s\}$.
Let us consider  two further sub-cases:
 $|S(\vec t+\vec z)\cap \text{supp}(\vec t)|=0$ or
  $|S(\vec t+\vec z)\cap \text{supp}(\vec t)|\geq 1$.
 
 For $|S(\vec t+\vec z)\cap \text{supp}(\vec t)|=0$, we trivially get 
\[ \sum_{g\in  \text{supp}(\vec u)} \kappa_g^*(\vec u+\vec z) \geq 0=\sum_{g\in  \text{supp}(\vec t)} \kappa_g^*(\vec t+\vec z).\]
For $|S(\vec t+\vec z)\cap \text{supp}(\vec t)|\geq 1$,
we get in the case $r\notin S(\vec t+\vec z)$
  \begin{align*} \sum_{g\in  \text{supp}(\vec u)} \kappa_g^*(\vec u+\vec z)&=
  \frac{B \cdot |S(\vec u+\vec z)\cap \text{supp}(\vec u)|}{|S(\vec u+\vec z)|}
  =  \frac{B \cdot (|S(\vec t+\vec z)\cap \text{supp}(\vec t)|+1)}{|S(\vec t+\vec z)|+1}
 \\ &\geq \frac{B \cdot (|S(\vec t+\vec z)\cap \text{supp}(\vec t)|)}{|S(\vec t+\vec z)|}
 =\sum_{g\in  \text{supp}(\vec t)} \kappa_g^*(\vec t+\vec z). 
  \end{align*}
  In case $r\in S(\vec t+\vec z)$, we have $|S(\vec t+\vec z)|=|S(\vec u+\vec z)|$ and
  $|S(\vec t+\vec z)\cap \text{supp}(\vec t)|=|S(\vec u+\vec z)\cap \text{supp}(\vec u)|$ and thus
\begin{align*} \sum_{g\in  \text{supp}(\vec u)} \kappa_g^*(\vec u+\vec z)&=
  \frac{B \cdot |S(\vec u+\vec z)\cap \text{supp}(\vec u)|}{|S(\vec u+\vec z)|}
  =  \frac{B \cdot (|S(\vec t+\vec z)\cap \text{supp}(\vec t)|)}{|S(\vec t+\vec z)|}
 \\ &
 =\sum_{g\in  \text{supp}(\vec t)} \kappa_g^*(\vec t+\vec z). \qedhere
  \end{align*}
 \end{itemize}
\end{proof}
This lemma together with the existence result in part~\ref{enum:2} of Theorem~\ref{polymatroid:main}
implies Theorem~\ref{thm:load-balancing}.

\section{Complexity of verifying equilibria}
\label{app:complexity}

\begin{theorem}
  It is $\mathsf{NP}$-hard to determine whether a given strategy profile of a resource graph game is a pure Nash equilibrium. This holds even when restricted to games where there is only a single player, the cost function $c$ fulfills the requirements of Theorem~\ref{thm:consistency_unweighted} and one of the following conditions is fulfilled:
  \begin{enumerate}
      \item $\mathcal{S}_1$ is a partition matroid on $R$.
      \item $\mathcal{S}_1$ is the set of $s$-$t$-paths in a directed graph and $|B_r| = 1$ for all $r \in R$.
  \end{enumerate}
  
\end{theorem}

\begin{proof}
  In case 1: We reduce from \textsc{3-SAT}. Given a \textsc{3-SAT} instance on $m$ clauses, let $z_{ij}$ denote the $j$th literal of clause $i$. We let $R = \{z_{ij} : i \in [m], j \in [3]\}$ be the set of resources of the game. The strategy space of the single player 1 is given by the bases of a partition matroid such that $S \subseteq R$ is a basis of the matroid if and only if $|S \cap \{z_{i1}, z_{i2}, z_{i3}\}| = 1$ for all $i \in [m]$. We further define the cost function $c_r(\vec x) = \sum_{s \in \Delta(r)} x_{s}$ where $\Delta(r)$ is the set of literals that contradict the literal $r$ (i.e., $s \in \Delta(r)$ if and only if $s$ and $r$ are literals of the same variable but with different signs).
  It is easy to see that player 1 has a strategy of cost $0$ if and only if there is a truth assignment fulfilling all clauses of the \textsc{3-SAT} instance.

  In case 2: We reduce from \textsc{Forbidden Pairs $s$-$t$-path}, which is known to be NP-hard~\cite{gabow1976two}: Given a digraph $G = (V, E)$, two nodes $s, t \in V$, and a collection of edge pairs $\{e_1, e'_1\}, \dots, \{e_k, e'_k\}$, does there exist an $s$-$t$-path $P$ such that $|P \cap \{e_i, e'_i\}| \leq 1$ for all $i \in [k]$? 
  
  We construct the resource graph game as follows: Let $R = E$ be the resource set. For $e \in E$ and $\vec x \in \mathbb{R}^{E}$ define $c_{e}(\vec x) = x_{e'_i}$ if $e = e_i$ for some $i \in [k]$, $c_{e}(x) = x_{e_i}$ if $e = e'_i$ for some $i \in [k]$ and $c_e(x) = 0$ otherwise. The game has a single player whose strategy space corresponds to the set of $s$-$t$-paths in $G$.
  It is easy to see that player 1 has a strategy of cost $0$ if and only if path avoiding all forbidden pairs.
\end{proof}

\bibliographystyle{plain}
\bibliography{master-bib}

\end{document}